\documentclass[11pt]{article}

\usepackage{amssymb} 
\usepackage{amsmath}     
\usepackage{amsthm}

\usepackage{hyperref}

\usepackage{a4wide}

\usepackage{graphicx}

\newtheorem{thm}{Theorem}

\newtheorem{cor}{Corollary}

\newtheorem{obs}{Observation}

\newcommand{\Qset}{\mathbb{Q}}

\newenvironment{keyword}{\par{\noindent\bf Keywords:}}

\begin{document}

\title{Approximability of the robust representatives selection problem}

\author{Adam Kasperski\\
   {\small \textit{Institute of Industrial}}\\
  {\small \textit{Engineering and Management,}}\\
  {\small \textit{Wroc{\l}aw University of Technology,}}\\
  {\small \textit{Wybrze{\.z}e Wyspia{\'n}skiego 27,}}\\
  {\small \textit{50-370 Wroc{\l}aw, Poland,}}\\
  {\small \textit{adam.kasperski@pwr.edu.pl}}
  \and  Adam Kurpisz, Pawe{\l} Zieli{\'n}ski\footnote{Corresponding author} \\
    {\small \textit{Institute of Mathematics}}\\
  {\small \textit{and Computer Science}}\\
  {\small \textit{Wroc{\l}aw University of Technology,}}\\
  {\small \textit{Wybrze{\.z}e Wyspia{\'n}skiego 27,}}\\
  {\small \textit{50-370 Wroc{\l}aw, Poland}}\\
{\small \textit{\{adam.kurpisz,pawel.zielinski\}@pwr.edu.pl}}
}

\date{}
\maketitle

\begin{abstract}
In this paper new complexity and approximation results on the robust versions of the representatives selection problem, 
under the scenario uncertainty representation,
are provided, which extend  the results obtained in the recent papers by Dolgui and Kovalev (2012), and Deineko and Woeginger (2013). Namely,
 it is shown that if the number of scenarios is a part of input, then 
  the min-max (regret) representatives selection problem
 is not approximable within a ratio of $O(\log^{1-\epsilon}K)$ for any $\epsilon>0$, where $K$ is the number of scenarios, unless the problems in NP have quasi-polynomial time algorithms. 
 An approximation algorithm with an approximation  ratio of $O(\log K/ \log \log K)$ for
 the min-max  version of the problem
  is also provided.
\end{abstract}


\section{Preliminaries}

In~\cite{DW13,DK12} the min-max and min-max regret versions of the following \emph{representatives selection} problem (\textsc{RS} for short) have been recently discussed. We are given a set $T$ of $n$ tools, numbered from $1$ to $n$. This set is partitioned into $p$ disjoint sets $T_1,\ldots,T_p$, where $|T_i|=r_i$ and $n=\sum_{i\in [p]}r_i$ 
(we use  $[p]$ to denote the set $\{1,\ldots p\}$).  
  We wish to choose a subset $X\subseteq T$ of the tools  that contains exactly one tool from each set $T_i$, i.e. 
 $|X\cap T_i|=1$ for each $i\in[p]$. In the deterministic case, each tool $j\in T$ has a nonnegative cost $c_j$ and we seek a solution~$X$ whose total cost $F(X)=\sum_{j\in X} c_j$ is minimal. This problem can be solved by a trivial algorithm, which chooses a tool of the smallest cost from each $T_i$. 
An important problem characteristic is the maximal number of elements in $T_i$, i.e. $r_{\max}=\max_{i\in [p]} r_i$.
We can assume that $r_i\geq 2$ for all $i\in[p]$, since a subset with the cardinality of~1 can be ignored. 
 It is easy to see that  \textsc {RS} is a special case of the \textsc{Shortest Path} problem, and this fact is depicted in Figure~\ref{fig1}. Each solid arc in the graph $G$ shown in Figure~\ref{fig1} corresponds to a tool in~$T_i$ and each dummy (dashed) arc has zero cost. There is one-to-one correspondence between the $s$-$t$ paths in~$G$ and the solutions to \textsc{RS}.   It is not difficult to transform the network $G$ and show that \textsc{RS} is also a special case of other basic network problems such as: \textsc{Minimum Spanning Tree}, \textsc{Minimum Assignment}, or \textsc{Minimum Cut} (see, e.g.~\cite{KZ09}).
\begin{figure}
\begin{center}
\includegraphics{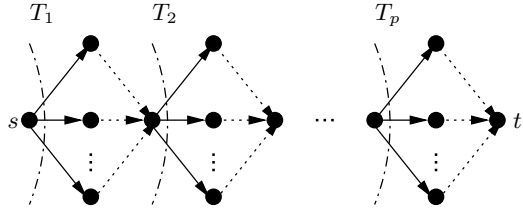}
\caption{The representative selection problem as the $s$-$t$ shortest path problem.}
\label{fig1}
\end{center}
\end{figure}

 Following~\cite{DK12}, let us now define the robust versions of \textsc{RS}.  Let $\Gamma=\{S_1,\dots,S_K\}$ be a \emph{scenario set}, where each \emph{scenario} is  a vector 
 $S_k=(c_{k1},\ldots,c_{kn})$ of nonnegative 
integral tool costs. Let $F(X,S_k)=\sum_{j\in X} c_{kj}$ be the cost of solution $X\in \Phi$
 under scenario $S_k$ and let $F^*(S_k)$ be the cost of an optimal solution under $S_k$. 
 In order to choose a solution,
two robust criteria,
called the \emph{min-max} and  the \emph{min-max regret}, can be
adopted (see~\cite{KY97} for a motivation of both robust criteria). 
In the \textsc{Min-Max RS} problem, we seek a solution which minimizes the largest cost over all scenarios, that is
\begin{equation}
OPT_1=\min_{X\in \Phi} cost_1(X)=\min_{X\in \Phi}\max_{k\in [K]} F(X,S_k). \label{opt1}
\end{equation}
In the \textsc{Min-Max Regret RS} problem, we wish to find a solution which minimizes the \emph{maximal regret}, that is
\[
 OPT_2=\min_{X\in \Phi} cost_2(X)=
\min_{X\in \Phi}\max_{k\in [K]} (F(X,S_k)-F^*(S_k)).
\]

We now describe the known results on  \textsc{Min-Max (Regret) RS}. Notice first that, all  positive results  for \textsc{Min-Max (Regret) Shortest Path} (see, e.g.~\cite{ABV09}) remain valid for \textsc{Min-Max (Regret) RS},
which is due to 
the fact that 
there is one-to-one correspondence between  $s$-$t$ paths in
 a layered digraph
and solutions to \textsc{RS} (see Figure~\ref{fig1}).
 Hence the latter problem admits a simple $K$-approximation algorithm which outputs an optimal solution for the aggregated tool costs $\hat{c}_j=\max_{k\in [K]} c_{kj}$, $j\in [n]$. 
 Furthermore, if $K$ is constant, then \textsc{Min-Max (Regret) RS} can be solved in pseudopolynomial time
  and admits an FPTAS. It has been shown in~\cite{DK12} that  \textsc{Min-Max (Regret) RS} is NP-hard even 
  when $K=2$ and $r_i=2$ for all $i\in [p]$, and becomes strongly NP-hard when the number of scenarios
   is a part of the input. This result has been recently extended in~\cite{DW13}, where it has been shown that when 
    $r_i=2$ for all $i\in [p]$, then  \textsc{Min-Max (Regret) RS} is as hard to approximate as the vertex cover problem, which is conjectured to be hard to approximate within $2-\epsilon$ for any $\epsilon>0$.  A similar inapproximability result, by assuming P$\neq$NP, follows immediately from the reduction given in~\cite{KZ09} for the \textsc{Min-Max (Regret) Shortest Path} problem. However, in the instances from~\cite{KZ09} we have 
     $r_i=3$ for all $i\in [p]$, so the result obtained in~\cite{DW13} is stronger.

\paragraph{Our results}
In this paper we investigate the case when both $K$ and $r_{\max}$ are 
parts of the input. We show that 
\textsc{Min-Max RS}
 admits an $r_{\max}$-approximation algorithm whose idea is to round 
up  solutions computed by solving a linear relaxation. In particular, it admits a 2-approximation algorithm when
 $r_{\max}=2$ which is the best possible according to the negative results given in~\cite{DW13}. 
 We also show, and these are the main results of the paper, that when additionally  $r_{\max}$ is a part of the input,  
 then \textsc{Min-Max (Regret) RS}
  is not approximable within $O(\log^{1-\epsilon} K)$ for any $\epsilon>0$ unless
   NP$\subseteq$DTIME$(n^{{\rm poly}(\log  n)})$ and provide
  an approximation algorithm with  a performance ratio of $O(\log K/\log\log K)$ for 
  \textsc{Min-Max RS},
  which is  close to the above lower bound.

\section{Hardness result}

In order to establish the hardness result, we will use the following variant of the \textsc{Label Cover} problem
(see e.g., \cite{AC95, MNO13}):

\begin{description}
\item[\mdseries \scshape Label Cover:]

We are given a regular  bipartite graph $G=(V\cup W,E)$, $E\subseteq V\times W$,
a set of labels $[N]$, and
for each edge $(v,w)\in E$ a  map (partial) $\sigma_{v,w}:[N]\rightarrow [N]$.
A \emph{labeling} of $G$ is an assignment of a subset of labels to each of the vertices of $G$, i.e. a function  $l: V\cup W \rightarrow 2^{[N]}$. We say that a labeling \emph{satisfies}
an edge $(v,w)\in E$ if there exist $a\in l(v)$ and $b\in l(w)$ such that $\sigma_{v,w}(a)=b.$
A \emph{total labeling} is a labeling that satisfies all edges.
We seek a  total labeling whose  value defined as  $\max_{x\in V\cup W}|l(x)|$ is minimal.
This minimal value is denoted by 
$val(\mathcal{L})$, where $\mathcal{L}$ is the input instance.

\end{description}

\begin{thm}[\cite{MNO13}]
There exists a constant~$\gamma>0$ such that
for any language $L\in NP$, any input $\mathbf{w}$ and any $N>0$,
one can construct a \textsc{Label Cover} instance~$\mathcal{L}$  with the following properties in time polynomial in the instance's size:
\begin{itemize}
\item the number of vertices in $\mathcal{L}$ is $|\mathbf{w}|^{O(\log N)}$,
\item if $\mathbf{w}\in L$, then $val(\mathcal{L})=1$,
\item if $\mathbf{w}\not\in L$, then $val(\mathcal{L})> N^{\gamma}$.
\end{itemize}
\label{tlancover}
\end{thm}
We now prove the following result:
\begin{thm}
There exists a constant~$\gamma>0$ such that
for any language $L\in NP$, any input $\mathbf{w}$, and any $N>0$,
one can construct an instance of
\textsc{Min-Max RS} with the following properties:
\begin{itemize}
\item if $\mathbf{w}\in L$, then $OPT_1\leq 1$,
\item if $\mathbf{w}\not\in L$, then $OPT_1\geq \lfloor N^\gamma \rfloor:=g$,
\item the number of tools is at most $|\mathbf{w}|^{O(\log N)}N$ and the number of scenarios is at most $|\mathbf{w}|^{O(g\log N)}N^{g}$.
\end{itemize}
\label{taprminmax}
\end{thm}
\begin{proof}
Let $L$ be a language in NP and let $\mathcal{L}=(G=(V\cup W,E),N,\sigma)$ be the instance of \textsc{Label Cover} constructed for $L$ (see Theorem~\ref{tlancover}). For each edge $(v,w)\in E$  we create a subset of tools $T_{v,w}$, which contains at most $N$ tools labeled as $(i,\sigma_{v,w}(i))$, $i\in [N]$.
A sample reduction is shown in Figure~\ref{fig2}.
\begin{figure}
\begin{center}
\includegraphics{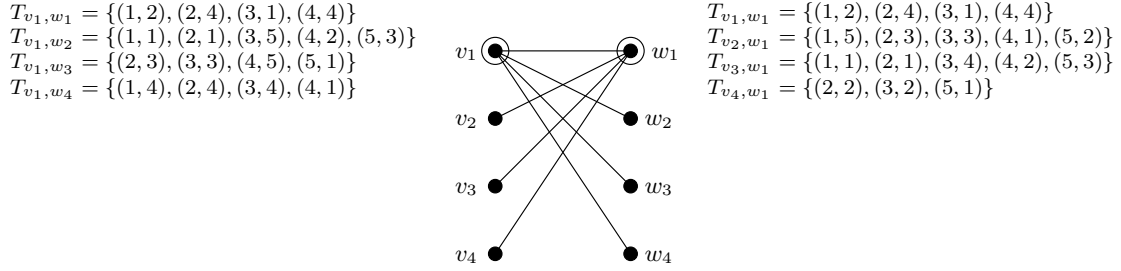}
\caption{A sample reduction for $K_{4,4}$ graph and $N=5$. Only the edges incident to $v_1$ and $w_1$ and the subsets of tools corresponding to them are shown.}
\label{fig2}
\end{center}
\end{figure}

A pair of tools $(i_1,j_1)\in T_{v_1,w_1}$ and $(i_2,j_2)\in T_{v_2,w_2}$ is \emph{label distinct} if $i_1=i_2$ implies $v_1\neq v_2$ and $j_1=j_2$ implies $w_1\neq w_2$. In other words, this pair cannot assign the same label twice to any vertex. For example, the pair $(3,1)\in T_{v_1,w_1}$ and $(1,1)\in T_{v_1,w_2}$ is label distinct, whereas  the pair $(1,2)\in T_{v_1,w_1}$ and $(1,1)\in T_{v_1,w_2}$ is not, as it assigns label 1 twice to $v_1$. We are now ready to form scenario set $\Gamma$. Let us fix a vertex $v\in V$. For each $g$-tuple $(v,w^1),\dots,(v,w^g)$ of pairwise distinct edges incident to $v$ and for each $g$-tuple of pairwise label distinct tools from $T_{v,w^1}\times T_{v,w^2}\times\ldots\times T_{v,w^g}$, we form a scenario under which these tools have costs equal to 1 and all the remaining tools have costs equal to 0. We repeat this procedure for each vertex $v\in V$. A sample scenario for vertex $v_1$ in Figure~\ref{fig2} and $g=3$ assigns the cost equal to 1 to tools $(1,2)\in T_{v_1,w_1}$, $(3,5)\in T_{v_1,w_2}$ and $(4,5)\in T_{v_1,w_3}$ as these tools forms a 3-tuple of pairwise label distinct tools. Let us fix a vertex $w\in W$. Now, for each $g$-tuple $(v^1,w),\dots,(v^g,w)$ of pairwise distinct edges incident to $w$ and for each $g$-tuple of pairwise label distinct tools from $T_{v^1,w}\times T_{v^2,w}\times\ldots\times T_{v^g,w}$, we form a scenario under which these tools have costs equal to 1 and all the remaining tools have costs equal to 0. We repeat this procedure for each vertex $w\in V$. A sample scenario for vertex $w_1$ in Figure~\ref{fig2} and $g=3$ assigns the cost equal to 1 to tools $(1,2)\in T_{v_1,w_1}$, $(1,5)\in T_{v_2,w_1}$ and $(3,4)\in T_{v_3,w_1}$ as these tools forms a 3-tuple of pairwise label distinct tools. To ensure that $\Gamma\neq \emptyset$, we create one additional scenario under which each tool has a cost equal to~0.

Suppose that $\mathbf{w}\in L$. Then $val(\mathcal{L})=1$ and there exists a total labeling $l$ which assigns one label $l(x)$ to each vertex $x\in V\cup W$. Let us choose tool $(l(v),l(w))\in T_{v,w}$ for each $(u,v)\in E$. We thus get a feasible selection of the tools $X$. It is easy to see that $F(X,S)\leq 1$ under each scenario $S\in \Gamma$ and, consequently $OPT_1\leq 1$. Assume that  $\mathbf{w}\notin L$. Then $val(\mathcal{L})>N^{\gamma}$ which implies $val(\mathcal{L})> \lfloor N^{\gamma}\rfloor=g$. Assume, by contradiction, that $OPT_1<g$, so there is a feasible selection $X$ such that $F(X,S)<g$ under each scenario $S\in \Gamma$. Observe that $X$ corresponds to the total labeling $l$ which assigns labels $i$ to $v_i$ and $j$ to $w_j$ when the tool $(i,j)$ is selected from $T_{v_i,w_j}$. From the construction of $\Gamma$ it follows that $l$ assigns less than $g$ distinct labels to each vertex $x\in V\cup W$, since otherwise $F(X,S)=g$ for some scenario $S\in\Gamma$. We thus have $val(\mathcal{L})< g$, a contradiction.

Let us now estimate the size of the constructed instance of \textsc{Min-Max RS} with respect to $|\mathbf{w}|$. The number of tools is  at most $|E| N$ and the number of scenarios is bounded by $|V||W|^gN^g+|W||V|^gN^g+1$.  According to Theorem~\ref{tlancover}, the number of vertices (and also edges) in $G$ is $|\mathbf{w}|^{O(\log N)}$. Hence the number of tools is at most $|\mathbf{w}|^{O(\log N)}N$ and the number of scenarios is at most $|\mathbf{w}|^{O(g\log N)}N^{g}$.
\end{proof}

\begin{thm} 
\textsc{Min-Max RS} is not approximable within $O(\log^{1-\epsilon} K)$, for any $\epsilon>0$, unless 
NP$\subseteq$DTIME$(n^{{\rm poly}(\log  n)})$.
\label{tngbmm}
\end{thm}
\begin{proof}
Let $\gamma$ be the constant from Theorem~\ref{taprminmax}. Consider a language $L\in$NP and an input~$\mathbf{w}$. Fix any constant $\beta>0$ and choose  $N=\lceil \log^{\beta/\gamma} |\mathbf{w}|\rceil$. According to Theorem~\ref{taprminmax}, we can  construct an instance of \textsc{Min-Max RS} in which the number of scenarios  $K$ is asymptotically bounded by $|\mathbf{w}|^{\alpha N^\gamma \log N}N^{N^\gamma}$ for some constant $\alpha>0$, $OPT_1\leq 1$ if $\mathbf{w}\in L$ and $OPT_1\geq\lfloor \log^\beta |\mathbf{w}| \rfloor$ if $\mathbf{w}\notin L$.
 We obtain $\log K \leq \alpha N^\gamma \log N \log|\mathbf{w}|+N^\gamma \log N\leq \alpha' \log^{\beta+2}|\mathbf{w}|$ for some constant $\alpha'>0$ and sufficiently large $|\mathbf{w}|$. Therefore, $\log|\mathbf{w}|\geq (1/\alpha') \log^{1/(\beta+2)}K$ and the gap is at least $\lfloor \log^\beta |\mathbf{w}| \rfloor\geq \lfloor 1/\alpha' \log^{\beta/(\beta+2)} K \rfloor$.
Since  the constant $\beta>0$ can be arbitrarily large, we get that the gap is $O(\log^{1-\epsilon} K)$ for any $\epsilon=2/(\beta+2)>0$.
Furthermore, the instance of \textsc{Min-Max RS} can be computed in $O(|\mathbf{w}|^{{\rm poly (log }|\mathbf{w}|)})$ time, which completes the proof.
\end{proof}

\begin{cor} 
\textsc{Min-Max Regret RS} is not approximable within $O(\log^{1-\epsilon} K)$, for any $\epsilon>0$, unless
NP$\subseteq$DTIME$(n^{{\rm poly}(\log  n)})$.
\end{cor}
\begin{proof}
A reduction is almost the same as the one from Theorem~\ref{taprminmax}. We only add to each $T_{u,v}$ a dummy tool with 0 cost under each scenario $S$ and one additional scenario $S'$ under which all the dummy tools have a large cost (say $g+1$) and the original tools have costs equal to 0. Thus, no dummy tool can be a part of an optimal solution and $OPT_2=OPT_1$.
\end{proof}

\section{Approximation algorithms}

In this section, we provide some LP-based 
approximation algorithms for  \textsc{Min-Max RS}.
Let us fix parameter~$L>0$ and let $T(L)\subseteq T$ be the set of all the tools $j\in T$ for which $c_{kj}\leq L$ for all scenarios $k\in [K]$. Clearly, $T(L)=\bigcup_{i\in[p]} T_i(L)$, where
$T_i(L)=\{j\in T_i: \forall_{k\in[K]} c_{kj}\leq L\}$ is a modified  subset of~$T_i$.
Consider the following linear program:
\begin{align}
\mathcal{LP}(L):&\sum_{j\in T(L)} c_{kj} x_j\leq L, &k\in [K],\label{sc}\\
                          &\sum_{j\in T_i(L)}x_j=1,& i\in[p],\label{hc}\\
                          &x_j\geq 0, & j\in T(L), \label{hc1} \\
													&x_j=0, & j\notin T(L).  \label{hc2}
\end{align}
Minimizing $L$ subject to~(\ref{sc})-(\ref{hc2}) we obtain an \emph{LP relaxation} of \textsc{Min-Max RS}.
Let  $L^*$ denote  the smallest value of the parameter~$L$
for which $\mathcal{LP}(L)$ is feasible.
Clearly, $L^*$ is a lower bound on $OPT_1$ and can be determined in polynomial time by using binary search. If $\pmb{x}^*$ is a feasible, fractional solution to  $\mathcal{LP}(L^*)$, then constraints~(\ref{hc}) imply $x_j\geq 1/r_{\max}$ for at least one tool $j \in T_i(L^*)$ for each $i\in [p]$. By choosing such a tool from $T_i(L^*)$, we get a 
solution~$X$ with $cost_1(X)$ at most $r_{\max}\cdot L^*\leq r_{\max}\cdot OPT_1$. This leads to the following observation:
\begin{obs}
 \textsc{Min-Max RS}
 admits an  $r_{\max}$-approximation algorithm.
\end{obs}
 It is worth pointing out that the above algorithm is the best possible when $r_{\max}=2$,  according to the
 negative results provided in~\cite{DW13}. The following observation describes the integrality gap of the LP relaxation:
\begin{obs}
	The LP relaxation has an integrality gap of
at least $\Omega(\log{K}/\log\log{K})$.
\end{obs}
\begin{proof}
 Let~$p>0$ be an arbitrary integer. 
Fix $n=p^2$. Let  $|T_i|=p$ for every $i\in [p]$. 
We form a scenario set~$\Gamma$ as follows. 
For each $p$-tuple~$(e_1,\dots,e_p)\in [p]^p$, where the component~$e_i\in [p]$ corresponds to the $e_i$-th element in subset~$T_i$,
we form scenario under which the tools indicated by  $(e_1,\dots,e_p)$
have costs equal to~$1$ and all the remaining tools have costs equal to~0.
Thus, $|\Gamma|=K=p^p$. Consider a fractional solution $x_j=1/p$, for all $j\in T$. This solution is feasible  to $\mathcal{LP}(L^*)$, where $L^*=1$.
It is easy to notice that each integer solution for the instance constructed has the maximum cost over all scenarios equal to $p$.
 Therefore, we have the integrality gap of~$p$. Since
$\log{K}=p \log{p}$, $\log{K}/\log\log{K} = p \log{p}/(\log{p}+ \log{\log{p}})=\Theta(p)$ and
 the integrality gap of the LP relaxation is  at least $\Omega( \log{K}/\log{\log{K}})$.
\end{proof}

We now convert  a feasible, fractional solution~$\pmb{x}^*$
 to  $\mathcal{LP}(L^*)$
 into an integer solution $\pmb{z}\in \{0,1\}^n$, which will represent a feasible tool selection. 
 To do this
  a randomized rounding technique proposed in~\cite{RT87,S99}
  can be applied.
   We first remove all the tools whose cost under some scenario is greater than $L^*$, i.e. $j\not \in  T(L^*)$.
We lose nothing by assuming, from now on, that the remaining tools~$j$, $j \in  T(L^*)$,  are
numbered from $1$ to $n$,  $|T(L^*)|=n$, and thus
$|T_i(L^*)|=r_i$ and $n=\sum_{j\in[p]}r_i$.
Furthermore,
we make the assumption that 
 $L^*=1$ and all the tool costs  are such that~$c_{kj}\in[0,1]$, $k\in[K]$, $j\in  [n]$. One can easily meet this
assumption by dividing  all the tool costs~$c_{kj}$ by $L^*$  for all $k\in[K]$ and  $j\in [n]$.
It turns out that
the above instance  of  the \textsc{Min-Max RS} problem with the scaled costs 
is an instance  of the \emph{minimax integer program} discussed in~\cite{S99}. Therefore,
we now construct $\pmb{z}$ as follows (see~\cite{RT87,S99}): 
independently, for each $i\in [p]$, pick one item $\overline{j}$ from $T_i$ with probability $x_{\overline{j}}^*$, set $z_{\overline{j}}=1$ and $z_j=0$ for the remaining items $j\in T_i$.
 It is clear that~$\pmb{z}$ represents a feasible tool selection. Let $\mathbf{C}$ denote the matrix corresponding to the scenario constraints~(\ref{sc}), i.e. $\mathbf{C}=(c_{kj}) \in [0,1]^{K\times n}$ and $(\mathbf{C} \pmb{z})_k=\sum_{j\in [n]} c_{kj} z_j$ stands for
the l.h.s value  of the scenario constraint~$k$
for solution~$\pmb{z}$. It holds $\mu_k=\mathrm{E}[(\mathbf{C} \pmb{z})_k]=(\mathbf{C} \pmb{x}^*)_k\leq L^*=1$ for all $k\in [K]$, due to
the linearity of expectation. We can make  the assumption that $\mu_k=1$  for all~$k\in [K]$, which can be satisfied
by adding slack variables to the constraints~(\ref{sc}). 
Now,
an analysis similar to that 
in~\cite[the proof of Theorem~6.1]{S99} shows that 
$(\mathbf{C} \pmb{z})_k$  is at most $O(\log K/\log\log K)$ for each $k\in [K]$ with high probability. 
We thus get a randomized  $O(\log K/\log\log K)$ approximation algorithm for the
\textsc{Min-Max RS} problem.

In order to derandomize the  above algorithm, one  can use
the \emph{method  of  pessimistic estimators}~\cite{R88}. However, this method  assumes that one can perform the computations with real numbers (in particular exponentials) with high precision in constant time.
Thus, it fails in the RAM model, when $\mathbf{C} \in [0,1]^{K\times n}$
(see~\cite{R88}, for comments). However, it works in  the RAM model when
 $\mathbf{C} \in \{0,1\}^{K\times n}$.
 Note that
 this particular case, i.e.
 \textsc{Min-max RS} with a binary matrix, is equivalent to the \emph{vector selection} problem 
studied in~\cite{R88}. 
We can thus formulate 
the analogue of results given in~\cite{R88}.
The following formulation, which provides  
 bounds on rounding errors, 
will be very  useful in a derandomization 
 of the randomized algorithm in the RAM model.
\begin{thm}
Let $\mathbf{C} \in \{0,1\}^{K\times n}$ and let $\pmb{x}\in ([0,1]\cap \Qset)^n$ be any 
solution that satisfies constraints~(\ref{hc}). Then
in the RAM model 
  a rounding $\pmb{z}\in\{0,1\}^n$ satisfying~(\ref{hc}) 
  can be computed,  by the method  of  pessimistic estimators
  in $O(Kn)$ time,
  such that the
 rounding error $|(\mathbf{C} \pmb{z})_k-(\mathbf{C}\pmb{x})_k|$ is $O(\max\{1,(\mathbf{C} \pmb{x})_k\} \log K/\log\log K)$
for all $k\in[K]$.
\label{thmder1}
\end{thm}
\begin{proof}
 The theorem can be proved in the same way as in~\cite[Theorems 5 and 6]{R88}.
\end{proof}
We now extend  Theorem~\ref{thmder1} to the case of $\mathbf{C} \in ([0,1] \cap \Qset)^{K\times n}$.
We make use of the derandomization  presented in~\cite{D06,D13}, which 
consists in approximating~$\mathbf{C}$ by
binary expansions 
and applying 
 the method  of  pessimistic estimators to
 binary matrices (in our case applying Theorem~\ref{thmder1}).
 \begin{thm}
Let $\mathbf{C} \in ([0,1] \cap \Qset)^{K\times n}$ and let $\pmb{x}\in ([0,1]\cap \Qset)^n$ be any 
solution that satisfies constraints~(\ref{hc}). Then
 a solution $\pmb{z}\in\{0,1\}^n$
 can be deterministically computed   in $O(Kn\log n)$ time such that  $\pmb{z}$ satisfies~(\ref{hc}) 
 and $(\mathbf{C} \pmb{z})_k$ is of value  $O(\max\{1,(\mathbf{C} \pmb{x})_k\} \log K/\log\log K)$
for all $k\in[K]$.
\label{deramth}
 \end{thm}
 \begin{proof}
The proof goes in the same manner as in~\cite[Theorem 2]{D13}. We  present it here 
 to show the idea and for the completeness.
 Set $t=\lceil \log n \rceil$ and express each entry of~$\mathbf{C}$ in $t$-bit binary form. This can be done in
  $O(Kn\log n)$ time.
 Let us denote by $\mathbf{\tilde{C}}$ the $t$-bits representation of~$\mathbf{C}$, namely
 $\mathbf{\tilde{C}}=\sum_{l\in[t]} 2^{-l}\mathbf{C}^{(l)}$, where $\mathbf{C}^{(l)}\in  \{0,1\}^{K\times n}$.
 Clearly, $|c_{kj}-\tilde{c}_{kj}|< 2^{-t}\leq 1/n$, $k\in [K], j\in [n]$.
 Thus, $||\mathbf{\tilde{C}} \pmb{x} - \mathbf{C} \pmb{x}||_{\infty}=||(\mathbf{\tilde{C}} -\mathbf{C}) \pmb{x}||_{\infty}\leq ||\mathbf{\tilde{C}} - \mathbf{C}||_{\infty} ||\pmb{x}||_{\infty}\leq 1$
 for every $\pmb{x}\in [0,1]^n$.
 We now apply Theorem~\ref{thmder1} to the solution~$\pmb{x}$ and 
  the $tK\times n$ binary matrix consisting of the rows of the binary matrices $\mathbf{C}^{(l)}$, $l\in[t]$
  and get in $O(tKn)$ time (and thus in $O(Kn\log n)$ time) in the RAM model a solution
  $\pmb{z}\in\{0,1\}^n$ satisfying~(\ref{hc}) 
   such that for all $k\in[K], l\in [t]$
  the rounding errors are as follows:
  \begin{equation}
   |(\mathbf{C}^{(l)} \pmb{z})_k-(\mathbf{C}^{(l)} \pmb{x})_k|\leq O(\max\{1,(\mathbf{C}^{(l)} \pmb{x})_k\} \log (tK)/\log\log (tK)).
   \label{ere}
  \end{equation}
  Therefore,
  \begin{align}
  |(\mathbf{\tilde{C}} \pmb{z})_k-(\mathbf{\tilde{C}} \pmb{x})_k|&=|\sum_{l\in[t]} 2^{-l}((\mathbf{C}^{(l)} \pmb{z})_k-
  (\mathbf{C}^{(l)} \pmb{x})_k)|\leq
  \sum_{l\in[t]} 2^{-l}|(\mathbf{C}^{(l)} \pmb{z})_k-(\mathbf{C}^{(l)} \pmb{x})_k|\nonumber\\
  &\overset{\text{(\ref{ere})}}{\leq}
  O(\log (tK)/\log\log (tK))\sum_{l\in[t]} 2^{-l}\max\{1,(\mathbf{C}^{(l)} \pmb{x})_k\} \nonumber\\
   &\leq O(\log (tK)/\log\log (tK))\sum_{l\in[t]} 2^{-l}(1+(\mathbf{C}^{(l)} \pmb{x})_k)\nonumber\\
   &\leq O(\log (tK)/\log\log (tK))(1+(\mathbf{\tilde{C}} \pmb{x})_k). \label{ere1}
  \end{align}
  Let us analyze a performance of the solution~$\pmb{z}$:
  \begin{align*}
  (\mathbf{C} \pmb{z})_k&\leq 1+(\mathbf{\tilde{C}} \pmb{z})_k\leq 1+(\mathbf{\tilde{C}} \pmb{x})_k+|(\mathbf{\tilde{C}} \pmb{z})_k-(\mathbf{\tilde{C}} \pmb{x})_k|
  \overset{\text{(\ref{ere1})}}{=} O(\max\{1,(\mathbf{\tilde{C}}\pmb{x})_k\}\log (tK)/\log\log (tK))\\
  & =O(\max\{1,(\mathbf{C} \pmb{x})_k\}\log (tK)/\log\log (tK)).
 \end{align*}
The theorem follows when $K\geq \log^{1/\gamma} n$ for some constant $\gamma\geq 1$, because it holds 
$t\leq K^\gamma$. We can satisfy this assumption by adding a number of zero cost (dummy) scenarios, if necessary.  Alternatively, we can handle the case with a small number of scenarios
(e.g.,   $K< \log^{1/\gamma} n$, for  $\gamma\geq 2$)
 by applying some algebraic techniques (see, e.g.,~\cite{D13}).
 \end{proof}
 From Theorem~\ref{deramth} and the fact that
$\pmb{x^*}$ is a feasible solution to $\mathcal{LP}(L^*)$,  we immediately  get the following corollary:
\begin{cor}
   \textsc{Min-Max RS} has a polynomial, deterministic 
   $O(\log K/\log\log K)$
approximation algorithm, where $K$ and $r_{\max}$ are parts of the input.  
\end{cor}

\section{Concluding remarks}

 There is still an open question concerning
the \textsc{Min-Max Regret RS} problem, when both~$K$ and $r_{\max}$ are parts of the input.
For this problem, there exists  a  $K$-approximation algorithm, known in the literature,
and $O(\log^{1-\epsilon} K)$ lower bound  on the approximability of  the problem,
given in this paper.
So, constructing a better approximation algorithm for
\textsc{Min-Max Regret RS} is an interesting subject of further
research.

\subsubsection*{Acknowledgements}
This work was 
partially supported by
 the National Center for Science (Narodowe Centrum Nauki), grant  2013/09/B/ST6/01525.


\end{document}